\documentclass{scrartcl}

\usepackage[affil-it]{authblk}

\usepackage{hyperref}

\usepackage{amsfonts}
\usepackage{amssymb,amsmath}
\usepackage{amsthm}

\newtheorem{theorem}{Theorem}
\newtheorem{corollary}{Corollary}         
\newtheorem{lemma}{Lemma}          
\newtheorem{proposition}{Proposition}

\newtheorem{remark}{Remark}

\usepackage{cite}

\newtheorem{conjecture}{Conjecture}

\newcommand{\pdbyd}[2][]{\frac{\partial #1}{\partial #2}}
\newcommand{\abs}[1]{\left\lvert #1 \right\rvert}

\newcommand{\naturals}{\ensuremath{\mathbb{N}}}
\newcommand{\reals}{\ensuremath{\mathbb{R}}}
\newcommand{\complexes}{\ensuremath{\mathbb{C}}}

\begin{document}


\title{Asymptotics of Certain Sums Required in Loop Regularisation
}

\author{Richard Chapling\thanks{Electronic address: 
\url{rc476@cam.ac.uk}}}

\affil{Department of Applied Mathematics and Theoretical Physics, \\ University of Cambridge, Cambridge, England}

\maketitle

\date{}

\begin{abstract}
We consider the three conjectures stated in a 2003 paper of Wu, concerning the asymptotics of particular sums of products of binomials, powers and logarithms. These sums relate to the form of the regularised integrals used in loop regularisation. We show all three are true, extend them to more general powers and produce their full asymptotic series. We also extend a classical result to produce an exact formula for the sum in the last.\\

\small
\emph{Keywords}: Loop regularization; asymptotic analysis.

PACS Nos.: 02.30.Mv., 11.15.Bt
\end{abstract}

\section{Introduction}

It is well-known\footnote{Weinberg \cite{Weinberg:1995kx}, Ch. 11; Folland \cite{Folland:2008vn}, Ch. 7; and Pauli's letters to Schwinger and Bethe\cite{Pauli:1993ys}} that regularisation is vital for valid physical interpretation of results produced by calculations in Quantum Field Theory; over the last half-century there have been many particular forms of regularisation suggested.\footnote{See Velo and Wightman \cite{velo:1976zr} for a broad survey thereof.}

One of the newest regularisation techniques is \emph{Loop Regularisation}, pioneered by Wu's paper of 2003\cite{ISI:000187232400003}. It possesses a number of technical advantages, including preservation of gauge symmetries, calculation being carried out in the original number of dimensions of the theory (rather than using an analytic continuation as in dimensional regularisation): this gives it significant advantage in chiral theories, where \( \gamma^{5} \) \emph{et al.} cause problems with other standard regularisation methods.\footnote{See \cite{Ma:2006fk} for an example where Dimensional Regularisation cannot correctly manage the \(\gamma^{5}\) terms.} Further, it has a certain simplicity of extra content added to the theory, in which two mass scales are introduced that act as soft cutoffs, and allow for natural theoretical incorporation of energy scales and the mass gap in various QFTs; regularising at the level of diagram integrals also avoids the introduction of further diagrams and bookkeeping particles, and this also allows it to maintain non-Abelian gauge invariance, unlike, say, Pauli--Villars regularisation.\footnote{\cite{Dai:2005ve,ISI:000258778900002,ISI:000293182400002}, and see in particular Wu's review article \cite{Wu:2014uq} for a summary of these.}

On the other hand, the theory also admits some interesting mathematical content: the consistency conditions derived in the second section of \cite{ISI:000187232400003} are of use by themselves in the calculation of scattering amplitudes for physical processes. In particular, they allow for the derivation of results independent of the specific regularisation scheme used.\footnote{See \cite{ISI:000301920000014} for an application to the \( H \to \gamma\gamma \) one-loop graphs, for example. Arguably this is the more theoretically interesting part of Wu's original paper, loop regularisation merely being a nice example of a scheme that satisfies these conditions.}

However, there was seen to be a small gap in the derivation of the simple forms of the \emph{irreducible loop integrals} (ILIs) key to the practical application of the theory: Wu's original paper assumes the veracity of three conjectures on the leading-order asymptotic behaviour of sums involved in the calculation of the regularised ILIs; this reduces the integrals to an easily computable form. These conjectures are essentially encompassed in equations (4.4), (4.7), and (4.9) of \cite{ISI:000187232400003}; we repeat them here for clarity, ease of access, and definition of our notation:
\begin{conjecture}
	Let $m$ be a positive integer. Then
	\begin{align}
	\label{eq:conj1}
		\sum_{k \geqslant 1} (-1)^{k}\binom{m}{k} k \log{k} \sim \frac{1}{\log{m}} \quad \text{as } m \to \infty.
	\end{align}
\end{conjecture}

\begin{conjecture}
	Let $m$ be a positive integer. Then
	\begin{align}
	\label{eq:conj2}
		\sum_{k \geqslant 1} (-1)^{k}\binom{m}{k} \log{k} \sim \log{\log{m}} + \gamma \quad \text{as } m \to \infty,
	\end{align}
	where $\gamma$ is the Euler--Mascheroni constant.
\end{conjecture}

\begin{conjecture}
	Let $m,n$ be positive integers. Then
	\begin{align}
	\label{eq:conj3}
		\sum_{k \geqslant 1} (-1)^{k-1}\binom{m}{k} \frac{1}{k^{n}} \sim \frac{\log^{n}{m}}{n!} \quad \text{as } m \to \infty.
	\end{align}
\end{conjecture}

To a certain extent, this omission is remedied by physical considerations and explicit construction in a subsequent paper of Wu\cite{ISI:000224011700003}, but since \cite{ISI:000187232400003} states ``Obviously, an analytical proof for the above conjectures must be very helpful and important. It may also provide deeper insights into mathematics'', the author felt that it would be worthwhile to provide such a rigorous proof: this is the purpose of this paper.

\section{A Useful Integral}

First we have a classical result due to Euler:\footnote{\cite{Euler:E368}, \S26. and \cite{Euler:E212}, Cap. I, \S16, see also \cite{Gould:1978db} for a more detailed modern discussion.}

\begin{theorem}
\label{thm:eulerdiff}
	Let $r$ be an integer with $0 \leqslant r<m$, and P(k) be a polynomial in $k$ of degree $r$. Then
	\[ \sum_{k \geqslant 0} (-1)^{k} \binom{m}{k} P(k) = 0. \]
\end{theorem}
This can be proved by considering the action of the differential operator $P(xD)$ on $(1-x)^{m}$, where $D=d/dx$.

\begin{remark}
	This result gives us an indication that the results expressed in the first two conjectures are plausible, since $k^{n} \log{k}$ is sandwiched between $k^{n}$ and $k^{n+1}$. Clearly the cancellation is very delicate, making numerical calculation of the series itself unreliable.
\end{remark}

Now define
\[ I(\alpha,m) = m \int_{0}^{\infty} y^{\alpha} e^{-y} (1-e^{-y})^{m-1}  \, dy. \]
The following lemma connects this integral to the sums we are considering:

\begin{lemma}
	Suppose $\alpha$ is complex and not a nonpositive integer, with $\Re{(m+\alpha)}>0$. Then $I(\alpha,m)$ exists and
	\begin{align}
	\label{eq:Inmnonintsum}
		I(\alpha,m) = \Gamma(\alpha+1) \sum_{k \geqslant 1} (-1)^{k-1} \binom{m}{k} \frac{1}{k^{\alpha}}.
	\end{align}
	If $\alpha=:-n \in \{-1,-2,\dotsc,m\}$, we instead have
	\begin{align}
	\label{eq:Inmintsum}
		I(-n,m) = \frac{(-1)^{n-1}}{(n-1)!} \sum_{k \geqslant 1} (-1)^{k} \binom{m}{k} k^{n} \log{k}.
	\end{align}
	Also,
	\begin{align}
	\label{eq:I'0sum}
		\pdbyd[I]{\alpha}(0,m) = \int_{0}^{\infty} e^{-y}(1-e^{-y})^{m-1}\log{y} \, dy = -\gamma - \sum_{k \geqslant 1}  (-1)^{k-1} \binom{m}{k} \log{k},
	\end{align}
	where $\gamma=-\Gamma'(1)$ is the Euler--Mascheroni constant.
\end{lemma}

\begin{proof}
	We shall prove these statements in the order given. For the integral to exist, notice that for $y>1$, the integrand is smaller than $y^{\Re{(\alpha)}}e^{-y}$. For $y<1$, it is instead bounded by $y^{\Re{(\alpha)}+m-1}$, so it follows that $I$ converges if $\Re{(m+\alpha)}>0$.
	
	Now suppose $\alpha>-1$, $m>0$. Then the integrand is bounded by $y^{\Re{(\alpha)}}e^{-y}$, so
	\begin{align*}
		I(\alpha,m) &= m\int_{0}^{\infty} y^{\alpha} \sum_{k \geqslant 1} (-1)^{k-1} \binom{m-1}{k-1} e^{-ky} \, dy \\
		&= \sum_{k \geqslant 1} (-1)^{k-1} \binom{m}{k} k \int_{0}^{\infty} y^{\alpha} e^{-ky} \, dy \\
		&= \sum_{k \geqslant 1} (-1)^{k-1} \binom{m}{k} \frac{\Gamma(\alpha+1)}{k^{\alpha}},
	\end{align*}
	where as $m>0$, the uniform absolute convergence of the sum allows us to swap the summation and integration. Both sides are meromorphic functions of $\alpha$, so they are in fact equal everywhere the integral exists, i.e. for $\Re{(m+\alpha)}>0$.
	
	Since the left-hand side of the previous equation is actually an analytic function of $\alpha$, the right-hand side possesses removable singularities at $\alpha=-1,-2,\dotsc,-(m-1)$.\footnote{Indeed, the vanishing of $I(\alpha,m)/\Gamma(\alpha+1)$ for $\alpha \in \{ -1, \dotsc, -(m-1) \}$ gives us another proof of Theorem \ref{thm:eulerdiff}.} To find the values at these points, we use L'H\^opital's rule: if $-n$ is a negative integer, for small $z$ we have
	\[ \Gamma(1-n+z) \sim \frac{(-1)^{n-1}}{(n-1)!z}. \]
	Also, 
	\[ \sum_{k \geqslant 1} (-1)^{k-1} \binom{m}{k} k^{n-z} = \sum_{k \geqslant 1} (-1)^{k-1} \binom{m}{k} k^{n} - z \sum_{k \geqslant 1} (-1)^{k-1} \binom{m}{k} k^{n} \log{k} + O(z^{2}) \]
	where the first term is equal to $0$ by Theorem \ref{thm:eulerdiff}, and hence the limit of the product as $z \to 0$ is
	\[ I(-n,m) = \frac{(-1)^{n-1}}{(n-1)!} \sum_{k \geqslant 1} (-1)^{k} \binom{m}{k} k^{n} \log{k}, \]
	which is the second formula.
	
	Finally, the first equality in the last formula is obvious. For the second, we need to find the derivative of the sum:
	\[ \pdbyd{\alpha}  \sum_{k \geqslant 1} (-1)^{k-1} \binom{m}{k} \frac{\Gamma(\alpha+1)}{k^{\alpha}} = \sum_{k \geqslant 1} (-1)^{k-1} \binom{m}{k} \frac{1}{k^{\alpha}}(\Gamma'(\alpha+1)-\Gamma(\alpha+1)\log{k}). \]
	Setting $\alpha=0$, Theorem \ref{thm:eulerdiff} again gives us
	\[ \pdbyd[I]{\alpha}(0,m) = - \gamma + \sum_{k \geqslant 1} (-1)^{k} \binom{m}{k} \log{k}. \]
\end{proof}

\section{Proof of the Conjectures and Calculation of the Full Asymptotic Series}

\subsection{Conjecture 1 and beyond}

Beginning with $I(-1,m)$, setting $e^{-u}=1-e^{-y}$, so $e^{-y} \, dy = - e^{-u} \, du$ and the limits swap, gives
\[ I(-1,m) = m\int_{0}^{\infty} \frac{e^{-mu}}{-\log{(1-e^{-u})}} \, du. \]
The integrand is positive, and it is easy to check that its derivative is decreasing on $(0,\infty)$. Consider splitting the integral at $c$, $0<c<1$. For $z>c$,
\[ \frac{e^{-mu}}{-\log{(1-e^{-u})}} < e^{-(m-1)u}, \]
since $-\log{(1-x)}>x$ for $0<x<1$. For $0<u<c$, see \ref{sec:ineq}. Then
\[ \int_{0}^{c} mu^{1-\epsilon}e^{-mu} du < \int_{0}^{\infty} mu^{1-\epsilon}e^{-mu} du = O(m^{-1+\epsilon}) \]

Therefore,
\begin{align*}
	\abs{I(-1,m) - m\int_{0}^{c}  \frac{e^{-mu}}{-\log{u}} \, du } &< A(c) \int_{0}^{c} mu^{1-\epsilon}e^{-mu} \, du + \int_{c}^{\infty} m e^{-(m-1)u} \, du \\
	&= O\left( m^{-(1-\epsilon)} \right);
\end{align*}
it shall follow that the asymptotic expansion of $I(-1,m)$ is in powers of $\log{m}$. To show this, we use the following theorem:\footnote{\cite{wong2001asymptotic}, p. 70, Theorem 2.}

\begin{theorem}
\label{thm:logasymp}
	For $\lambda,\mu,c \in \complexes$, $\Re{(\lambda)}>0$, $c = \abs{c}e^{i\gamma}$, $0<\abs{c}<1$, the integral
	\begin{align}
		\label{eq:genlogint}
		L(\lambda,\mu,z) = \int_{0}^{c} t^{\lambda-1} (-\log{t})^{\mu} e^{-z t} dt
	\end{align}
	has asymptotic expansion
	\[ L(\lambda,\mu,z) \sim z^{-\lambda} (\log{z})^{\mu} \sum_{r=0}^{\infty} (-1)^{r} \binom{\mu}{r} \Gamma^{(r)}(\lambda) (\log{z})^{-r} \]
	uniformly in $\arg{z}$, as $\abs{z} \to \infty$ in $\abs{\arg{(ze^{i\gamma})}} \leqslant \pi/2-\Delta<\pi/2$.
\end{theorem}

The proof is straightforward and uses the binomial expansion after substituting $x=zt$. From this we conclude that $I(-1,m)$ has asymptotic expansion
\[ I(-1,m) \sim \frac{1}{\log{m}} \sum_{r=0}^{\infty} (-1)^{r} \binom{-1}{r} \Gamma^{(r)}(1) (\log{m})^{-r}, \]
a considerable improvement of Conjecture 1.

However, we also know that for any given positive integer $n$, for $m$ sufficiently large we have
\[ \sum_{k \geqslant 1} (-1)^{k-1} \binom{m}{k} k^{n} \log{k} = (-1)^{n}(n-1)! I(-n,m). \]
Carrying out the same substitution as above gives
\[ I(-n,m) = m\int_{0}^{\infty} \frac{e^{-mu}}{(-\log{(1-e^{-u})})^{n}} \, du.  \]
Again splitting the integral at $c$, we have for $u>c$,
\[ \frac{e^{-mu}}{(-\log{(1-e^{-u})})^{n}} < e^{-(m-n)u}, \]
so
\[ m\int_{c}^{\infty} \frac{e^{-mu}}{(-\log{(1-e^{-u})})^{n}} \, du < \frac{m}{m-n}e^{(m-n)c} = O(e^{-mc}) \]
as $m \to \infty$. \ref{sec:ineq} gives that the difference on $(0,c)$ is bounded by
\begin{align*}
	\int_{0}^{c}me^{-mu}&\abs{(-\log{(1-e^{-u})})^{-n}-(-\log{u})^{-n}} \, du  \\
	&< A_{1} m\int_{0}^{c}u^{1/2}(-\log{u})^{-n-1/2} e^{-mu} \, du \\
	&< A_{2} m\int_{0}^{c}u^{(1-\epsilon)/2} e^{-mu} \, du \\
	&= O\left(m^{-(1/2-\epsilon/2)} \right),
\end{align*}
where the $A_{i}$ are constants. Therefore the asymptotic series is given by Theorem \ref{thm:logasymp}:
\[ \sum_{k \geqslant 1} (-1)^{k} \binom{m}{k} k^{n} \log{k} \sim \frac{(-1)^{n-1}(n-1)!}{(\log{m})^{n}} \sum_{r=0}^{\infty} (-1)^{r} \binom{-n}{r} \Gamma^{(r)}(1) (\log{m})^{-r}. \]

\subsection{Conjecture 2}
It is plain that this should work in a similar fashion: we begin with \eqref{eq:I'0sum}. Once again using the substitution $e^{-u}=1-e^{-y}$, the integral becomes
\[ m \int_{0}^{\infty} e^{-mu} \log{\left(-\log{(1-e^{-u})}\right)} \, du =: K(m). \]
The integrand is no longer always positive: it is positive for $u<\log{(1-e^{-1})}$ and negative subsequently. Splitting the integral at $c$ as before, we see that for large $u$,
\begin{align*}
	\abs{\log{\left(-\log{(1-e^{-u})}\right)} + u} &= \abs{\log{\left(-e^{u}\log{(1-e^{-u})}\right)}} \\
	&< \abs{\log{\left(-e^{u}(-e^{-u}+e^{-2u}/2)\right)}} \\
	&= \abs{\log{\left(1-\tfrac{1}{2}e^{-u}\right)}} < B(c),
\end{align*}
for some constant $B(c)$, so
\[ \int_{c}^{\infty} me^{-mu} \log{\left(-\log{(1-e^{-u})}\right)} \, du < \int_{c}^{\infty} m(u+B(c))e^{-mu} \, du = O(me^{-mc}). \]
Similarly, for $u<c$, the results of \ref{sec:ineq} give
\begin{align*}
	\abs{\log{\left(-\log{(1-e^{-u})}\right)} - \log{(-\log{u})}} = o(u^{1-\epsilon}),
\end{align*}
and so we have, for some constant \(A(c)\),
\begin{align*}
	\abs{K(m) - m\int_{0}^{c} e^{-mu} \log{\left(-\log{u}\right)} \, du} &< A(c)m\int_{0}^{c}  u^{1-\epsilon} e^{-mu} \, du + O(me^{-mc}) \\
	&= O\left( m^{-(1-\epsilon)} \right).
\end{align*}

Therefore the expansion is again given by a series in the logarithm: we have the similar result\footnote{\cite{wong2001asymptotic}, p. 71} that the integral
\begin{align}
	\label{eq:genloglogint}
	F(\lambda,z) = \int_{0}^{c} t^{\lambda-1} \log{(-\log{t})} e^{-z t} dt
\end{align}
has asymptotic expansion
\[ F(\lambda,z) \sim z^{-\lambda} \Gamma(\lambda) \log{\log{z}} - z^{-\lambda} \sum_{r=1}^{\infty} \frac{1}{r}\Gamma^{(r)}(\lambda) (\log{z})^{-r} \]
uniformly in $\arg{z}$, as $\abs{z} \to \infty$ in $\abs{\arg{(ze^{i\gamma})}} \leqslant \pi/2-\Delta<\pi/2$.

We therefore conclude:

\[ \sum_{k \geqslant 1} (-1)^{k} \binom{m}{k} \log{k} \sim \log{\log{m}} + \gamma - \sum_{r=1}^{\infty} \frac{1}{r}\Gamma^{(r)}(1) (\log{m})^{-r}, \]
a considerable improvement of Conjecture 2.

\subsection{Conjecture 3}
Lastly, we carry out the same procedure for $I(\alpha,m)$ where $\alpha \in \complexes \setminus \{ 0, -1, -2, \dotsc \}$. The same substitution as before gives
\[ I(\alpha,m) = m \int_{0}^{\infty} e^{-mu} \left(-\log{(1-e^{-u})}\right)^{\alpha} \, du. \]

We split the integral at $c$, $0<c<1$. The $u>c$ term is bounded by
\[ \frac{m}{n} \int_{c}^{\infty} e^{-(m-\Re{(\alpha)})z} dz = O\left( e^{-mc} \right) = o(m^{-1}) \]
as $m \to \infty$. As before, for $u<c$, the results of \ref{sec:ineq} show that
\[ \abs{\left(-\log{(1-e^{-u})}\right)^{\alpha} - \left(-\log{u}\right)^{\alpha}} = O\left( u^{1/2} (-\log{u})^{\Re{(\alpha)}-1/2} \right), \]
and because for any $1>\epsilon>0$, $u(-\log{u})^{\Re{(\alpha)}-1/2} = o(u^{1-\epsilon})$ as $u \to 0$, we have
\begin{align*}
	&\abs{I(\alpha,m)-m \int_{0}^{c} e^{-mu} \left(-\log{u}\right)^{\alpha} \, du } \\
	&\quad< \frac{m}{2}\int_{0}^{c} u^{1/2}(-\log{u})^{\Re{(\alpha)}-1/2} e^{-mu} du + o(m^{-1}) \\
	&\quad= o(m^{-(1/2-\epsilon)}),
\end{align*}
so the expansion is again given by a special case of \eqref{eq:genlogint}'s expansion:
\[ \sum_{k \geqslant 1} (-1)^{k} \binom{m}{k} k^{-\alpha} \sim -\frac{(\log{m})^{\alpha}}{\Gamma(\alpha+1)} \sum_{r=0}^{\infty} (-1)^{r} \binom{\alpha}{r} \Gamma^{(r)}(1) (\log{m})^{-r}, \]
which is a considerable improvement of Conjecture 3.

\section{Exact formulae for Conjecture 3 with $n,m \in \naturals$}
Euler\footnote{\cite{Euler:E726}, \S12; we have changed the notation to be consistent with our own.} gives for $m \in \naturals$ the following formula:
\[ \sum_{k \geqslant 1} (-1)^{k}\binom{m}{k} \frac{1}{k} = -\sum_{k=1}^{m} \frac{1}{k} = -H_{m}, \]
the $m$th harmonic number. He proves this by examining the series expansion of both sides of
\[ \frac{z^{c}}{(1-z)^{c+1}} \log{\left(1+\frac{z}{1-z}\right)} = \frac{z^{c}}{(1-z)^{c+1}} \log{(1-z)}, \]
and considering the coefficients of $z^{m}$ in the special case $c=0$.

It is well-known that
\begin{align}
\label{eq:harm1asymp}
	H_{m} = \log{m} + \gamma + O\left( \frac{1}{m} \right),
\end{align}
which proves the $n=1$ case of Conjecture 3. Is there a generalisation of this to larger values of $n$? Yes, although the expressions become more complicated.

For $m>0$, $s<0$, we set
\[ J_{m}(s) = m\int_{0}^{\infty} e^{-y}(1-e^{-y})^{m-1} e^{sy} \, dy; \]
we observe that
\[ J_{m}(s) = \sum_{k \geqslant 1} (-1)^{k-1} \binom{m}{k} \frac{k}{k-s} = \sum_{n = 0}^{\infty} s^{n} \sum_{k \geqslant 1} (-1)^{k-1} \binom{m}{k} \frac{1}{k^{n}} =: \sum_{n = 0}^{\infty} s^{n} S_{m,n}, \]
so $J_{m}(s)$ is in fact a generating function for the $S_{m,n}$. The substitution $x=e^{-y}$ transforms $J_{m}(s)$ into a beta integral:
\[ J_{m}(s) = m\int_{0}^{1} (1-x)^{m-1} x^{-s} \, dx = \frac{\Gamma(m+1)\Gamma(1-s)}{\Gamma(m+1-s)} \]
which is an analytic function of $s$ which does not vanish in a neighbourhood of the origin. Hence $\log{J_{m}(s)}$ exists and has a convergent power series.

To relate this to the generalised harmonic numbers,
\[ H^{(n)}_{m} = \sum_{k=1}^{m} \frac{1}{k^{n}}, \]
recall the functional equation for the $\Gamma$-function,
\[ \Gamma(z+1) = z \Gamma(z). \]
Logarithmically differentiating this equation gives
\[ \digamma(z+1) = \frac{1}{z} + \digamma(z), \]
where $\digamma$ is the Digamma-function $(\log{\Gamma})'$, and so
\[ H_{m} = \digamma(m+1) - \digamma(1). \]
Differentiating a further $n-1$ times gives
\[ \digamma^{(n-1)}(z+1) = \frac{(-1)^{n-1}(n-1)!}{z^{n}} + \digamma^{(n-1)}(z), \]
and it follows that
\[ H^{(n)}_{m} = \frac{(-1)^{n-1}}{(n-1)!}( \digamma^{(n-1)}(m+1) - \digamma^{(n-1)}(1) ). \]
Therefore,
\begin{align}
\label{eq:logBetaharm}
	\log{\Gamma(m+1)}+\log{\Gamma(1-s)}&-\log{\Gamma(m+1-s)} \nonumber \\
	&= \sum_{n=1}^{\infty} \frac{(-1)^{n-1}}{n!}( \digamma^{(n-1)}(m+1) - \digamma^{(n-1)}(1) ) s^{n} \nonumber \\
	&= \sum_{n=1}^{\infty} \frac{H^{(n)}_{m}}{n} s^{n}
\end{align}

Of course, now we have to express the derivatives of the functions we actually want in terms of these. For this we use the exponential formula:\footnote{\cite{Stanley:1999le}, p.5, Corollary 5.1.6} suppose
\[ f(x) = \sum_{n=1}^{\infty} a_{n} \frac{x^{n}}{n!}.  \]
Then
\[ e^{f(x)} = \sum_{n=0}^{\infty} b_{n} \frac{x^{n}}{n!}, \]
where
\[ b_{n}= \sum_{\{A_{1},\dotsc,A_{k}\} \in \Pi(n)} a_{\abs{A_{1}}} \dotsm a_{\abs{A_{k}}}, \]
where $\Pi(n)$ is the set of partitions of $\{ 1,\dotsc,n \}$, so, for example,
\[ \Pi(3) = \bigg\{ \big\{ \{1\},\{2\},\{3\} \big\} , \big\{ \{1,2\},\{3\} \big\} , \big\{ \{2,3\},\{1\} \big\} , \big\{ \{3,1\},\{2\} \big\} , \big\{ \{1,2,3\} \big\} \bigg\}, \]
and hence $b_{3}=a_{1}^{3}+3a_{1}a_{2}+a_{3}$.\footnote{These are also known as the (complete) Bell polynomials $Y_{n}(a_{1},\dotsc,a_{n})$.}

We now apply this to $ a_{n} = (n-1)! H^{(n)}_{m} $, as we found in (\ref{eq:logBetaharm}); we find
\begin{align*}
	J_{m}(s) = \exp{\left( \log{\Gamma(m+1)}+\log{\Gamma(1-s)}-\log{\Gamma(m+1-s)} \right)} &= \sum_{n=0} b_{n} \frac{s^{n}}{n!}
\end{align*}
with $b_{n}$ as discussed above.\footnote{The absolute convergence of this series is guaranteed by the convergence of the series expansions of $f$ and $\exp$, but we can also manage by manipulating them as formal power series in $s$.} Extracting the coefficient of $s^{n}$, we conclude:
\begin{theorem}
\label{thm:sum3exactform}
	Suppose $m$, $n$ are positive integers. Then
	\begin{align}
		\sum_{k \geqslant 1} (-1)^{k-1} \binom{m}{k} \frac{1}{k^{n}} = \frac{1}{n!} \sum_{\{A_{1},\dotsc,A_{k}\} \in \Pi(n)} a_{\abs{A_{1}}} \dotsm a_{\abs{A_{k}}},
	\end{align}
	where $a_{n} = (n-1)! H^{(n)}_{m}$.
\end{theorem}

However, there is more we can extract from this formula: most importantly from our point of view, because $\forall n > 1$, we have the asymptotic\footnote{This is easily shown using the Euler--Maclaurin formula, for example, or just the integral test if so desired.}
\[ H^{(n)}_{m} = \zeta(n) + O(m^{1-n}), \quad m \to \infty. \]
Therefore any and all divergent contributions to the series come from $H_{m}$, with the asymptotic we gave in \eqref{eq:harm1asymp}. We can use these facts to extract as many terms in the asymptotic series as we like; note that it is clear that all the divergent terms form a polynomial in $\log{m}$. Theorem \ref{thm:sum3exactform} therefore gives us the following elementary resolution of Conjecture 3:

\begin{corollary}
	Let $n$ be a positive integer. Then
	\[ \sum_{k \geqslant 1} (-1)^{k-1} \binom{m}{k} \frac{1}{k^{n}} = \frac{\log^{n}{m}}{n!} + \frac{\gamma\log^{n-1}{m}}{(n-1)!} + O(\log^{n-2}{m}) \]
\end{corollary}

\begin{proof}
	To obtain the given terms, we need to consider the terms involving the highest powers of $H_{m}$; it is apparent that these are confined to the $(H_{m})^{n}$ term, because any partition of $\{ 1, \dotsc , n \}$ either contains $n$ sets of size $1$, or a set of size $2$ or larger, in which case it can only contain a maximum of $(n-2)$ sets of size $1$. Therefore the leading order terms are
	\[ \frac{1}{n!}(H_{m})^{n} = \frac{1}{n!}\left( \log{m}+\gamma +O(1/m) \right)^{n} = \frac{1}{n!}\left(\log^{n}{m} + n\gamma \log^{n-1}{m}\right) + O(\log^{n-2}{m}). \]
	The error term follows automatically from the above discussion.
\end{proof}

\begin{remark}
	It is in fact not necessary to restrict $m$ to positive integers for the above calculation, since the polygamma functions have identical asymptotics to the harmonic numbers.
\end{remark}

\section{Discussion}

Our integral $I(\alpha,m)$ is considerably easier to study asymptotically than the sum itself: this is true both of producing the asymptotic expansion and of numerical calculation; the heavy cancellation that occurs in the sums with positive powers in particular is especially hard to manage.

Due to the presence of the logarithms, we do not expect the positive powers to have an exact formula.

It is curious that positive powers cancel so well, whereas negative powers grow; the negative powers are explained by the generalised harmonic sums discussed above, but positive powers' behaviour lacks an easy justification. The $\alpha=0$ case is typically borderline.

\section{Conclusion}
We have proven the conjectures of Wu,\cite{ISI:000187232400003} and extended the results considerably to include the full asymptotic series, as well as a more general case for positive powers larger than $1$. We have also produced an exact formula for the negative-integer power case. While we do not think this has necessarily provided ``deep insights into mathematics'', we have no doubt that the reduced complexity of the expressions that result, and the implications for the structure of the theory, will be a relief and aid to physicists wishing to use and study LORE.

\appendix

%
%
%
%

\section{Inequalities}
\label{sec:ineq}

The most difficult parts of applying Theorem \ref{thm:logasymp} are bounding the extra contributions from the interval $(0,c)$ and the entire integral on $[c,\infty)$; in this section we shall prove some inequalities that aid this.

For large $u$, we require the following inequalities for the :

\begin{lemma}[Logarithm inequalities]
\label{thm:logineq}
	For $x<1$,
	\[ x \leqslant -\log{(1-x)} \leqslant \frac{x}{1-x}. \]
\end{lemma}

\begin{proof}
	In some senses this is elementary: $e^{y}$ lies above its tangent at $0$, so
	\[ y \leqslant e^{y}-1; \]
	substituting $y=-\log{(1-x)}$ gives
	\[ -\log{(1-x)} \leqslant \frac{1}{1-x}-1 = \frac{x}{1-x}, \]
	and substituting $y=\log{(1-x)}$ gives
	\[ \log{(1-x)} \geqslant -x.  \]
\end{proof}

\begin{corollary}
	Suppose $u \geqslant c>0$ and $a > 0$. Then
	\[  e^{-au} \leqslant (-\log{(1-e^{-u})})^{a} \leqslant \frac{e^{-au}}{(1-e^{-c})^{a}}; \]
	the inequality is reversed if $a<0$.
\end{corollary}
\begin{proof}
	Set $x=e^{-u}$ in the previous lemma: this gives
	\[ e^{-u} \leqslant -\log{(1-e^{-u})} \leqslant \frac{e^{-u}}{1-e^{-u}}. \]
	The function $(1-e^{-u})^{-1}$ is decreasing for $u \geqslant 0$, so we find
	\[ e^{-u} \leqslant -\log{(1-e^{-u})} \leqslant \frac{e^{-u}}{1-e^{-c}} \]
	for $u>c$. Now, for $a>0$ ($<0$), $x^{a}$ is increasing (decreasing), so the sense of the inequalities is preserved (reversed) by raising them to the power $a$, which gives the result.
\end{proof}

The more difficult is the result required for small $u$: first we require some simple inequalities:

\begin{lemma}
\label{thm:taylorineq}
	Suppose $f \in C^{n+1}[a,b]$, and $f^{(n+1)}(x) \geqslant 0$ on $[a,b]$. Then for any $a \leqslant x \leqslant b$,
	\[ f(x) \geqslant \sum_{k=0}^{n} \frac{f^{(k)}(a)}{k!}. \]
\end{lemma}
\begin{proof}
	Use Taylor's theorem with, for example, the Lagrange form of the remainder,
	\[ \frac{f^{(k+1)}(\xi)}{(k+1)!}(x-a)^{k+1}, \quad a \leqslant \xi \leqslant x  \]
	and note that this is positive.
\end{proof}

(The reverse case of $f^{(n+1)}(x) \leqslant 0$ is a trivial application of the lemma to $-f$.)

\begin{lemma}[Binomial inequalities]
\label{thm:binomineq}
	Let $n \in \reals$. Then for every $x \in (0,1)$,
	\[ x\min{\{n,2^{n}-1\}} < (1+x)^{n}-1 < x\max{\{n,2^{n}-1\}}. \]
\end{lemma}

\begin{proof}
	We shall prove the upper bounds; the lower bounds are proved in exactly the same way. There are two cases to check. Define $f(x)=(1+x)^{n}-1$, and consider $f''(x)$:
	\[ f''(x) = n(n-1)(1+x)^{n-2}; \]
	the last bracket is plainly positive, so $f''(x)$ is uniformly nonpositive for $n \in [0,1]$, and nonnegative otherwise. Suppose first that $f''(x)$ is nonpositive. Then the reverse of the previous lemma implies
	\[ f(x) \leqslant nx. \]
	Now suppose that $f''(x)>0$. Then $f$ is convex, so the graph lies below the $(0,1)$ secant and
	\[ f(x) \leqslant (2^{n}-1)x; \]
	it is easy to check that $2^{n}-1>n$ except in $[0,1]$, and the theorem follows.
\end{proof}

\begin{remark}
	This is a special case of the following: suppose $f:[a,b] \to \reals$ has $f'(x)$ nondecreasing for every $x \in [a,b]$. Then
	\[ f'(a) \leqslant \frac{f(x)-f(a)}{x-a} \leqslant \frac{f(b)-f(a)}{b-a},  \]
	i.e. the graph lies between the secant and the tangent. For $f'(x)$ nonincreasing, the inequality is reversed.
	
	The proof is easy: using the mean value theorem and that $f'$ is nondecreasing gives the left inequality, and the right follows from the definition of convexity.
\end{remark}

Finally, the rest is done with
\begin{lemma}
\label{thm:absbound}
	Let $\alpha \in \complexes$, and write $\alpha = a + ib$, $a,b \in \reals$. Then for $0 \leqslant X < C < 1$,
	\[ \abs{1-(1+X)^{\alpha}} \leqslant A X^{1/2}, \]
	$A>0$ dependent only on $\alpha$ and $C$.
\end{lemma}

\begin{proof}
	Obviously the sensible thing to do is square the left-hand side and work from there. Recalling that $z^{a} = e^{a\log{z}}$, and using $\abs{x+iy}^{2} = x^{2}+y^{2} $, we find
	\[ \abs{1-(1+X)^{\alpha}}^{2} = 1 -2(1+X)^{a}\cos{b \log{(1+X)}} + (1+X)^{2a}. \]
	The bounds from the previous lemma now come into play, along with the well-known inequalities
	\[ -\cos{x} \leqslant \frac{1}{2}x^{2}-1, \qquad 0 \leqslant \log{(1+X)} \leqslant X \]
	to give
	\begin{align*}
		\abs{1-(1+X)^{\alpha}}^{2} &\leqslant 1+2(1+m X)(\tfrac{1}{2}b^{2}X^{2}-1) + (1+M X) \\
		&= (M-2m)X + X(b^{2} X+b^{2} m X^{2}),
	\end{align*}
	where $m = \min\{ a,2^{a}-1 \}$ and $M = \max\{ 2a, 2^{2a}-1 \}$. Clearly the second term is positive for sufficiently small $X$; the first term can also be shown to be positive by case-by-case consideration, and the result follows.
\end{proof}

We can now lump all these lemmata together to prove:

\begin{proposition}
\label{thm:powerbound}
	Suppose $u<c<1$. Then for any $\alpha \in \complexes$,
	\[ \abs{(-\log{(1-e^{u})})^{\alpha}-(-\log{u})^{\alpha}} \leqslant B(c) u^{1/2}(-\log{u})^{\Re(\alpha)-1/2} \]
	for some $B(c)$.
\end{proposition}

\begin{proof}
	Write $\alpha = a + ib$, $a,b \in \reals$. First notice that both $-\log{(1-e^{u})}$ and $-\log{u}$ are positive real numbers; the significance of this is that for positive reals $X$ and $Y$,
	\[ \abs{X^{\alpha}} = X^{a}  \]
	and
	\[ X^{a}Y^{a} = (XY)^{a}. \]
	Therefore we rewrite the left-hand side of the inequality as
	\[ \abs{(-\log{(1-e^{u})})^{\alpha}-(-\log{u})^{\alpha}} = (-\log{u})^{a} \abs{1 - \left( \frac{-\log{(1-e^{u})}}{-\log{u}} \right)^{a} }. \]
	We also have
	\[ -\log{(1-e^{-u})} \leqslant -\log{(u-u^{2}/2)} \leqslant -\log{u} + \tfrac{1}{2}u, \]
	using 
	\[ 1-e^{-u} \geqslant u-\tfrac{1}{2}u^{2}, \]
	so for sufficiently small $c$,
	\[ \frac{-\log{(1-e^{u})}}{-\log{u}} - 1 \leqslant C \]
	and we can apply the previous lemma to obtain
	\begin{align*}
		(-\log{u})^{a} &\abs{1 - \left( \frac{-\log{(1-e^{u})}}{-\log{u}} \right)^{a} } \\
		&\leqslant A(\alpha,C(c)) (-\log{u})^{a} \left( \frac{-\log{(1-e^{-u})}-(-\log{u})}{-\log{u}} \right)^{1/2} \\
		&\leqslant A(a,C(c)) (-\log{u})^{a} \left( \frac{u}{-2\log{u} } \right)^{1/2} \\
		&\leqslant A(a,C(c)) B u^{(1-\varepsilon)/2} (-\log{u})^{a} ;
	\end{align*}
	the final inequality comes from $-\log{u}< B u^{-\varepsilon}$ on $(0,c]$ for some $B>0$ and a small enough $c$.
\end{proof}

(The result can probably be strengthened to $O(u^{1-\varepsilon})$, but we do not need this.)

On the other hand, the integral for Conjecture 2 requires the much simpler

\begin{lemma}
For $0<u<c<1$ and some $D>0$,
	\[ \abs{\log{(-\log{(1-e^{-u})})} - \log{(-\log{u})} } \leqslant D u^{1-\varepsilon}. \]
\end{lemma}

\begin{proof}
	\begin{align*}
		\abs{\log{(-\log{(1-e^{-u})})} - \log{(-\log{u})} } &= \abs{\log{\left( \frac{-\log{(1-e^{-u})} }{-\log{u} } \right) } } \\
		&\leqslant \frac{-\log{(1-e^{-u})} - (-\log{u}) }{-\log{u} } - 1 \\
		\intertext{using $\log{x}\leqslant x-1$,}
		&\leqslant \frac{u}{2(-\log{u})} \\
		&\leqslant D u^{1-\varepsilon},
	\end{align*}
	using that $\log{u} = o(u^{-\varepsilon})$ as $u \to 0$.
\end{proof}

\section*{Acknowledgements}
I would like to thank Stephen Siklos for suggesting improvements to my approach that considerably improved the estimate, as well as Piers Bursill-Hall, Imre Leader and David Stuart for their help and support.

\end{document}